\documentclass[pra,twocolumn,amsmath,amssymb,nofootinbib]{revtex4}

\usepackage[bookmarks = false, pdfpagemode = None, pdfstartview = FitH, colorlinks = true, urlcolor = blue,hyperfootnotes=false]{hyperref}
\usepackage{amsthm}
\usepackage{eufrak}
\usepackage{color}

\renewcommand{\i}{\mathrm{i}}
\newcommand{\e}{\mathrm{e}}
\DeclareMathAlphabet{\mathpzc}{OT1}{pzc}{m}{it}
\newcommand{\set}[1]{\mathcal{#1}}
\newcommand{\code}[1]{\mathsf{#1}}
\newcommand{\system}[1]{\mathpzc{#1}}
\newcommand{\liealg}[1]{\mathfrak{#1}}

\newtheorem{lem}{Lemma}
\newtheorem{thm}{Theorem}
\newtheorem{cor}{Corollary}

\begin{document}
\title{Restrictions on Transversal Encoded Quantum Gate Sets}

\author{Bryan Eastin}
\email{beastin@nist.gov}

\author{Emanuel Knill}

\affiliation{National Institute of Standards and Technology,
Boulder, CO 80305}

\begin{abstract}
Transversal gates play an important role in the theory of fault-tolerant quantum computation due to their simplicity and robustness to noise.  By definition, transversal operators do not couple physical subsystems within the same code block.  Consequently, such operators do not spread errors within code blocks and are, therefore, fault tolerant.  Nonetheless, other methods of ensuring fault tolerance are required, as it is invariably the case that some encoded gates cannot be implemented transversally.  This observation has led to a long-standing conjecture that transversal encoded gate sets cannot be universal.  Here we show that the ability of a quantum code to detect an arbitrary error on any single physical subsystem is incompatible with the existence of a universal, transversal encoded gate set for the code.
\end{abstract}

\maketitle

Quantum computation appears to be intrinsically more powerful than its
classical counterpart. Efficient quantum algorithms have been found
for certain problems that, using the best known classical algorithms,
require resources that scale as a super-polynomial function of the problem
size~\cite{Shor94,Aharonov06,Mosca08}.
However, implementing a computation large enough to take advantage of such scaling properties is a daunting challenge.  Given the difficulty of constructing quantum hardware, it seems likely that the software for the first quantum computers will need to incorporate significant amounts of error checking.

As in the classical case, quantum errors are rendered detectable by
encoding the system of interest into a subspace of a larger, typically
composite, system. A quantum code simply specifies which states of a
quantum system correspond to which logical (encoded) information
states.
Errors that move states outside of the logical subspace can be
detected by measuring the projector $P$ onto this subspace.  Thus, an error $E$ is detectable, in the sense that it can be discovered or eliminated, if and only if
\begin{align*}
  P E P \propto P\;.
\end{align*}
Of course, not all errors can be detected; for any nontrivial code there are operators that act in a nontrivial way within the logical subspace.  Most commonly, quantum codes are designed to permit the detection of independent, local errors and, as a consequence, are incapable of detecting some errors that affect many subsystems.

For quantum computation, it is necessary not only to detect errors but also to apply operators (gates) that transform the logical state of the code.  Even when error processes are local and independent, however, the operations entailed in computing can generate correlated errors from uncorrelated ones.  Thus, for error detection to be
effective, it is important that the logical operators employed during a
quantum computation be designed to limit the spread of errors.
It is particularly important that operators do not spread errors within code blocks, where a block of a quantum code is defined as a collection of subsystems for which errors on subsystems in the collection are detected independently of those on subsystems outside of it.
Managing the spread of errors is the subject of the theory of fault-tolerant quantum computing~\cite{Shor97,Preskill98}.
One of the primary techniques of this theory is the use of transversal encoded gates.

We label as ``transversal'' any partition of the physical subsystems of a code such that each part contains one subsystem from each code block.
Given a transversal partition of a code, an operator is called transversal if it exclusively couples subsystems within the same part.
Put another way, an operator is transversal if it couples no subsystem of a
code block to any but the corresponding subsystem in another code block.
Transversal operators are inherently fault tolerant. They can spread
errors between code blocks, thereby increasing the number of locations
at which a code block's error might have originated, but, since
errors on different code blocks are treated independently, the total number of errors
necessary to cause a failure is unchanged. This is in contrast to
non-transversal operators, where, for example, an encoded gate
coupling every subsystem in a code block might convert an error on a
single subsystem into an error on every subsystem of the code block.

In view of the above, it would be highly desirable to carry out
quantum computations exclusively using transversal encoded gates. To allow for
arbitrary computation, it is necessary that the set of gates employed be
universal, that is, that it be capable of implementing any encoded
operator on the logical state space to arbitrarily high accuracy. However, in
spite of substantial effort, no gate set for a nontrivial
quantum code has yet been found that is both universal and
transversal.  Consequently, a long-standing question in quantum information theory
is whether there exist nontrivial quantum codes for which all logical
gates can be implemented transversally.  For stabilizer codes, this question has recently
been answered in the negative.  Zeng et al.~\cite{Zeng07} showed that transversal unitary operators are not universal for stabilizer codes on two-level subsystems (qubits); the companion result for the
case of $d$-level subsystems (qudits) was proven by Chen et al.~\cite{Chen08}.
In this paper we present a more general proof based on the structure of the Lie
group of transversal unitary operators.
Our result applies to all
local-error-detecting quantum codes, that is, all quantum codes capable of detecting an arbitrary error on any single subsystem.

An outline of the argument is as follows: The set of logical unitary product operators, $\set{G}$, is a Lie subgroup of the Lie group of unitary product operators, $\set{T}$.  As a Lie group, $\set{G}$ can be
partitioned into cosets of the connected component
of the identity, $\set{C}$; these cosets form a discrete set, $\set{Q}$.  Using the fact that the Lie
algebra of $\set{C}$ is a subalgebra of $\set{T}$, it can be shown that the connected component of
the identity acts trivially for any local-error-detecting code. This implies that the number of logically distinct operators implemented by elements of $\set{G}$
is limited to the cardinality of $\set{Q}$.
Due to to the compactness of $\set{T}$, this
number must be finite. A finite number of operators can approximate infinitely many
only up to some fixed
accuracy; thus, $\set{G}$, the set of logical unitary product operators, cannot be
universal.
Transversal operators may be viewed as product operators with respect to a transversal partitioning of the code, so the ability to detect an arbitrary error on a transversal part implies the non-existence of a universal, transversal encoded gate set.

We begin by exploring the structure of various sets of unitary operators and
subsequently move to our central theorem.  The following material relies heavily on results from topology and the theory of Lie groups.
An accessible introduction to these topics can be found, for example, in Refs.~\cite{Munkres,Hall} and on Wikipedia~\cite{Wikipedia}.

Consider a quantum system of finite dimension $d$. The set
$\set{U}(d)$ of unitary operators on a $d$-dimensional quantum
system forms a compact, connected Lie group with a Lie algebra
consisting of the Hermitian operators.\footnote{Following the convention in physics, we include a factor of
$\i$ in the mapping between elements of the Lie algebra and Lie
group.}
% See Hall pg. 11 & 15 for compactness and connectedness of the group of unitaries.
Thus, any unitary operator $U\in\set{U}(d)$ satisfies
\begin{align*}
  U = \e^{\i H}
\end{align*}
for some Hermitian operator $H$.
% Zelobenko pg. 291 states that any element of a compact, connected Lie group can be written as a single exponential (rather than as a product of exponentials) of an element of the Lie algebra.

Now consider a composite quantum system $\system{Q}$ composed of $n$
physical subsystems, where the dimension of the $j$th subsystem is
$d_j$. Let $\set{T}$ denote the set of all unitary product operators, that is, all operators
of the form
\begin{align*}
  \bigotimes_{j=1}^n U_j\;,
\end{align*}
where $U_j\in \set{U}(d_j)$.  Being a
direct product of a finite number of compact Lie groups, $\set{T}$ is
also a compact Lie group. For the same reason, $\set{T}$ has a Lie
algebra $\liealg{t}$ given by the direct sum of the Lie algebras of the
component groups.

Given a quantum code $\code{C}$ on the system $\system{Q}$, the set of logical unitary operators on $\system{Q}$ is defined as the subset of unitary operators that preserve the code space.  In terms of a projector $P$ onto the code states of
$\code{C}$, this is the statement that a unitary operator $U$ is a logical operator if and only if
\begin{align}\label{eq:defLogicalOperator}
  (I-P) U P = 0 \;.
\end{align}
Note that $(I-P) U P$ is a
continuous function of $U$.

\begin{lem}\label{lem:logicalGroup}
The set of logical unitary operators forms a group.
\end{lem}
\begin{proof}
Let $P$ be the projector onto the logical subspace of a quantum code.  The set of logical unitary operators, $\set{L}$, consists of all unitary operators $U$ satisfying
\begin{align*}
  P U P = U P\;.
\end{align*}
The set $\set{L}$ fulfills the four requirements of a group:
The multiplication of unitary operators is associative.  The identity, $I$, is contained in $\set{L}$ as
\begin{align*}
  P I P = P^2 = P = I P\;.
\end{align*}
The group property of closure is satisfied since
\begin{align*}
  PUVP = PUPVP = UPVP = UVP
\end{align*}
for any $U,V\in\set{L}$. The inverse $U^\dagger$ of any $U\in
\set{L}$ is contained in $\set{L}$ since
\begin{align*}
  \left(P U^\dagger P\right) \left(P U P\right) = \left(P U^\dagger\right) \left(U P\right) = P\;,
\end{align*}
which implies that $P U^\dagger P$ is the inverse of $P U P$ on the
subspace $P$ and therefore that
\begin{align*}
U^\dagger (P) = U^\dagger (P U P P U^\dagger P) = U^\dagger U P P U^\dagger P = P U^\dagger P\; .
\end{align*}
\end{proof}

\begin{lem}\label{lem:logicalLieGroup}
The logical operators contained in a Lie group of unitary operators form a Lie subgroup.
\end{lem}
\begin{proof}
Let $\set{L}$ be the set of logical unitary operators for a given code, let $\set{A}$ be a Lie group of unitary operators, and let $\set{B}=\set{A}\cap \set{L}$.  Lemma~\ref{lem:logicalGroup} shows that $\set{L}$ is a group.  Because the intersection of two groups is a group, $\set{B}$ is a subgroup of $\set{A}$.  Topologically speaking, $\set{L}$ is a closed set since, as seen from Eq.~\ref{eq:defLogicalOperator}, it is a preimage of a closed set under a continuous function.  Being a Lie group, $\set{A}$ is also a topologically closed set, and therefore $B$ is as well.  That $\set{B}$ is a Lie subgroup of $\set{A}$ follows from a theorem by Cartan (See page 3 of Ref.~\cite{Sepanski}.), which states that a topologically closed subgroup of a Lie group is a Lie subgroup.
\end{proof}

\begin{thm}\label{thm:productGatesAndErrorDetection}
For any nontrivial local-error-detecting quantum code, the set of logical unitary product operators is not universal.
\end{thm}

\begin{proof}
\ \par

Let $\system{Q}$, as defined earlier, be a composite quantum system
supporting a local-error-detecting code $\code{C}$.  The set of
unitary product operators on $\system{Q}$ is the compact Lie group that
was earlier denoted by $\set{T}$.

Lemma~\ref{lem:logicalLieGroup} shows that $\set{G}$, the subset of unitary product operators that are also logical
operators, forms a Lie subgroup of $\set{T}$.

As a Lie group, $\set{G}$ can be partitioned into cosets of the connected component of the identity, $\set{C}$, where $\set{C}$ is a Lie subgroup of $\set{G}$.  This set of cosets is the quotient group $\set{Q} = \set{G}/\set{C}$ and constitutes a topologically discrete group.

Because $\set{C}$ is a connected Lie group, any element $C\in\set{C}$ can be written as
\begin{align*}
  C = \prod_k \e^{\i D_k}\;,
\end{align*}
where $D_k$ is in $\liealg{c}$, the Lie algebra of $\set{C}$.
For any $D\in\liealg{c}$ and $\epsilon\in \Re$, the operator $\e^{\i\epsilon
D}$ is also in $\set{C}$ and is, consequently, a logical gate satisfying
\begin{align*}
  0 = (I-P) \e^{\i\epsilon D} P\;.
\end{align*}
Since $(I-P)IP = 0$, we also have
\begin{align*}
  0 = \lim_{\epsilon\rightarrow 0} (I-P) \left[\frac{\e^{\i\epsilon D}-I}{\i\epsilon}\right] P = (I-P)D P
\end{align*}
for all $D\in\liealg{c}$.

As $\set{C}$ is a Lie subgroup of the Lie group $\set{T}$, its Lie
algebra $\liealg{c}$ must be a subalgebra of $\liealg{t}$, the Lie algebra
of $\set{T}$. Consequently, every element $D\in\liealg{c}$ can
be written in the form
\begin{align*}
  D = \sum_{j=1}^n \alpha_{j}H_j\;,
\end{align*}
where $\alpha_{j}\in\Re$ and $H_j$ is a Hermitian operator applied to the $j$th subsystem.
Any local Hermitian operator can be written as a sum over local error operators, so
\begin{align*}
  P H_j P \propto P\;,
\end{align*}
where $P$ is the projector onto the code space of $\code{C}$.

Combining the preceding three equations yields
\begin{align*}
  \begin{split}
    D P = P D P = P \sum_{j=1}^n \alpha_{j}H_j P
    = \sum_{j=1}^n \alpha_{j} P H_j P \propto P
  \end{split}
\end{align*}
for all $D\in\liealg{c}$, which shows that
\begin{align*}
  C P = \prod_k \e^{\i D_k} P \propto P \;.
\end{align*}
Since $C$ is a unitary operator, the constant of proportionality must be one.
Thus, whether it is trivial or not, all operators contained in
$\set{C}$ act as the identity on the code space.

Let $\set{F}$ be a set consisting of one representative from each coset of $\set{C}$ in $\set{G}$.
The preceding paragraph shows
that every operator in the group $\set{G}$ acts on the code space as
an operator from $\set{F}$.  In other words, for every
$G\in\set{G}$,
\begin{align*}
  G P = F C P = F P
\end{align*}
for some $F\in\set{F}$ and $C\in\set{C}$.

The operators induced by $\set{G}$ on the logical quantum system are closed under composition and limited in number to the cardinality of $\set{F}$.
The set $\set{F}$ is discrete since its elements are representatives taken from each of the cosets comprising the discrete group $\set{Q}=\set{G}/\set{C}$.
It follows that $\set{F}$ is also finite, being a discrete subset of a compact
group, namely $\set{T}$.  However, for a nontrivial encoded quantum
system, the number of logically distinct operators is uncountably infinite.
As the set of all unitary operators is a metric space, a finite number of unitary operators cannot approximate infinitely many
to arbitrary precision.\footnote{By contrast, the Solovay-Kitaev theorem~\cite{Kitaev97b,Dawson06} states that a universal, and infinite, set of operators can be generated by composition from certain finite sets of operators.  In our case, composition yields nothing new.}
Thus, $\set{G}$, the set of logical product operators, is not universal.
\end{proof}

Theorem~\ref{thm:productGatesAndErrorDetection} considers only
product gates, but the same basic approach
can be applied to the case of transversal gates.

\begin{cor}\label{cor:transversalGatesAndErrorDetection}
For any nontrivial local-error-detecting quantum code, the set of transversal, logical unitary operators is not universal.
\end{cor}

\begin{proof}
This result follows directly from an application of
Theorem~\ref{thm:productGatesAndErrorDetection} in which the physical subsystems are replaced by transversal parts.  Each part contains a set of physical subsystems that can be coupled by transversal operators.  Transversal operators may therefore be regarded as product operators on the transversal parts.  Theorem~\ref{thm:productGatesAndErrorDetection} thus proves that the set of transversal, logical unitary operators is not universal for any nontrivial quantum code capable of detecting an arbitrary error on a single transversal part.  For a local-error-detecting code,
the condition that any error on
a single transversal part be detectable is satisfied since this corresponds to a
single-subsystem error on each of the code blocks.
\end{proof}

As with any impossibility proof, perhaps the most interesting aspect
of Corollary~\ref{cor:transversalGatesAndErrorDetection} is how it can be circumvented. The most obvious
circumvention, and an avenue that has been thoroughly explored, is to
employ non-unitary operators~\cite{Knill04,Bravyi05,Zhou00}.  The standard method of achieving universal fault-tolerant quantum computation
takes this approach, making extensive use of measurements and classical feed-forward during the
preparation, testing, and coupling of ancillary states.  Alternatively, one might retain unitarity and instead loosen the requirements
of transversality or universality or even error detection,
options that we discuss in turn.

Among the alternatives listed, non-transversal operators provide the most promising approach to circumventing Theorem~\ref{thm:productGatesAndErrorDetection}.
References~\cite{Zeng07,Chen08} discuss the possibility of achieving
universality through the addition of coordinate permutations, which,
taken in isolation, are fault tolerant.  Zeng et al.\ note that the
encoded Hadamard gate for the Bacon-Shor codes~\cite{Bacon06} involves a coordinate
permutation and therefore is not transversal.  In fact, for these codes, some sequences of
encoded Hadamard and controlled-NOT gates are not fault tolerant; a single physical gate failure is capable of producing
two errors on a single code block.
Strict fault tolerance is achieved by checking for errors prior to coupling code blocks using a new transversal partition.  Such codes demonstrate that it
is sufficient for individual logical gates to avoid directly
coupling subsystems of a code block.  A quantum code for which there existed a universal set of encoded gates each transversal in isolation would be extremely useful.

Along a different line, we might imagine demanding less than full universality.  Finite groups of operators are already an important component of schemes for fault-tolerant quantum computing.  These schemes typically take advantage of the existence of codes for which the Clifford gates, a finite subgroup of all gates, are both sufficient for error detection and transversally implementable.  The Clifford gates are not the only set that can be implemented transversally, however.  It would be interesting to quantify the maximum size of finite group
that is achievable transversally and to investigate the computational power of the non-Clifford finite gate groups.

Given a local error model, it seems unprofitable to abandon
local error detection entirely.  In order to violate the assumptions of our proof, however, it is sufficient that detection not be deterministic.  It might be possible to find a family of codes satisfying both the universality and transversality conditions for which the probability of failing to detect an error on a single subsystem can be made arbitrarily small.  The usefulness of such a family of codes would depend on the scaling of the failure probability with the size of the code.

In conclusion, we have presented a proof that the ability of a quantum
code to detect arbitrary errors
on component subsystems is incompatible with the existence of a
universal, transversal, and unitary encoded gate set. Our proof makes
no assumptions about the dimensions of the quantum subsystems beyond
requiring that they be finite.  The quantum system encoded is assumed to
be nontrivial, that is, to have dimension greater than one.  The precise structure of the quantum code and its initialization state are unspecified.
Our result rules out the use of transversal unitary operators with local error detection as an exclusive means to
obtain universality, but it also suggests some interesting new avenues of
investigation.\\

\acknowledgments

We thank Adam Meier, Scott Glancy, and Yanbao Zhang for their questions and
comments during the development of this proof. Special thanks go to
Sergio Boixo for the discussion that spawned the idea that local error
detection and infinitesimal transversal gates were incompatible.  Preliminary investigations on this
topic were funded by National Science Foundation Grant
No.~PHY-0653596. This paper is a contribution by the National
Institute of Standards and Technology and, as such, is not subject to
US copyright.

\bibliography{../../citations}

\end{document}